\documentclass[aps,pra,reprint,groupedaddress,superscriptaddress,twocolumn,notitlepage,nofootinbib
]{revtex4-2}
\pdfoutput=1

\usepackage{dcolumn}
\usepackage{physics,amsmath,amsthm,amsfonts,graphicx,times,xfrac,mathtools,amssymb,bm,verbatim,appendix,mathrsfs,scalerel,lipsum}
\usepackage[normalem]{ulem}
\usepackage{xcolor,hyperref}
\theoremstyle{plain}
\newtheorem{definition}{Definition}

\newtheorem{theorem}{Theorem}

\newtheorem{example}{Example}
\newtheorem{claim}{Claim}

\newcommand{\Def}{\coloneqq}
\newcommand{\Id}{\Tilde{\mathbb{I}}}
\newcommand{\Ibb}{\mathbb{I}}
\newcommand{\Sbb}{\mathbb{S}}
\newcommand{\Pbb}{\mathbb{P}}
\newcommand{\Tbb}{\mathbb{T}}

\newcommand{\rmA}{\mathrm{A}}
\newcommand{\rmB}{\mathrm{B}}
\newcommand{\atb}{\rmA\to\rmB}
\newcommand{\rmAB}{\mathrm{AB}}
\newcommand{\rmC}{\mathrm{C}}
\newcommand{\rmD}{\mathrm{D}}

\newcommand{\rmR}{\mathrm{R}}

\newcommand{\rmX}{\mathrm{X}}

\newcommand{\Acal}{\mathcal{A}}

\newcommand{\Ccal}{\mathcal{C}}
\newcommand{\Dcal}{\mathcal{D}}
\newcommand{\Ecal}{\mathcal{E}}

\newcommand{\Hcal}{\mathcal{H}}
\newcommand{\Ical}{\mathcal{I}}

\newcommand{\Mcal}{\mathcal{M}}
\newcommand{\Ncal}{\mathcal{N}}
\newcommand{\Pcal}{\mathcal{P}}
\newcommand{\Rcal}{\mathcal{R}}
\newcommand{\Scal}{\mathcal{S}}
\newcommand{\Tcal}{\mathcal{T}}

\newcommand{\Vcal}{\mathcal{V}}

\newcommand{\Zcal}{\mathcal{Z}}

\begin{document}

\title{Process tensor distinguishability measures}

\author{Guilherme Zambon}
\affiliation{
 Instituto de F{\'i}sica de S{\~a}o Carlos, Universidade de S{\~a}o Paulo, CP 369, 13560-970, S{\~a}o Carlos, SP, Brasil
 }
 \affiliation{
 School of Mathematical Sciences and Centre for the Mathematical and Theoretical Physics of Quantum Non-Equilibrium Systems, University of Nottingham, University Park, Nottingham, NG7 2RD, United Kingdom
 }

\begin{abstract}
  Process tensors are quantum combs describing the evolution of open quantum systems through multiple steps of quantum dynamics. While there is more than one way to measure how different two processes are, special care must be taken to ensure quantifiers obey physically desirable conditions such as data-processing inequalities. Here, we analyze two classes of distinguishability measures commonly used in general applications of quantum combs. We show that the first class, called Choi divergences, does not satisfy an important data-processing inequality, while the second one, which we call generalized divergences, does. We also extend to quantum combs some other relevant results of generalized divergences of quantum channels. Finally, given the properties we proved, we argue that generalized divergences may be more adequate than Choi divergences for distinguishing quantum combs in most of their applications. Particularly, this is crucial for defining monotones for resource theories whose states have a comb structure, such as resource theories of quantum processes and resource theories of quantum strategies.
\end{abstract}

\maketitle

\section{Introduction}

Distinguishability measures play a key role in several applications of quantum information theory. Besides their use in discrimination tasks, they are also employed for defining information and correlation quantifiers, which include a plethora of monotones for quantum resource theories, both static and dynamic \cite{chitambar2019quantum,gonda2023monotones}. While there is no single way to define such measures, one must ensure that the proposed quantities satisfy physically desirable properties, like additivity for independent systems, non-negativity, continuity, and contractivity under noisy operations \cite{yuan2019hypothesis,gour2021entropy}. This last property, also known as data-processing inequality, is especially important when defining monotones for resource theories, as it usually implies that any distinguishability measure with respect to the closest free state will be monotonic under the free operations of the theory \cite{chitambar2019quantum,gonda2023monotones}.

While static and dynamical resource theories have been largely studied, one could also consider manipulations of more general objects. In the resource theory of quantum processes, for example, the objects are process tensors, which consist of quantum combs that operationally describe the evolution of open quantum systems interacting with an experimenter multiple times during the dynamics \cite{chiribella2009theoretical,pollock2018non,Berk2021resourcetheoriesof}. The main advantage of using process tensors to this end is the fact that they offer a proper definition of quantum Markovianity, allowing for the investigation of purely quantum properties of information flow in general quantum processes \cite{pollock2018operational,milz2019completely,taranto2019quantum,taranto2019structure,figueroa2019almost,milz2020kolmogorov,milz2020when,milz2021genuine,taranto2021non,figueroa2021markovianization,sakuldee2022connecting,capela2022quantum,taranto2023hidden,taranto2024characterising,zambon2024relations,santos2024quantifying}. For this reason, process tensors have been used to understand the role of non-Markovianity in several fields of quantum theory that had so far mostly been studied under the Markov assumption, like quantum thermodynamics \cite{strasberg2019repeated,figueroa2020equilibration,huang2022fluctuation,huang2023multiple,dowling2023relaxation,dowling2023equilibration}, quantum process tomography \cite{milz2018reconstructing,white2020demonstration,white2022non,white2022characterization,white2023filtering,aloisio2023sampling}, and quantum simulation \cite{jorgensen2019exploiting,jorgensen2020discrete,cygorek2022simulation,fowler2022efficient,gribben2022using,cygorek2024sublinear,fux2023tensor}, among others \cite{figueroa2021randomized,figueroa2022towards,berk2023extracting,figueroa2024operational,butler2024optimizing}.

When assessing the influence of non-Markovianity on other relevant properties of quantum systems, it is, in general, useful to define non-Markovianity quantifiers. The way this is usually done is through state distinguishability measures between the Choi states of process tensors and the closest Markovian Choi state \cite{pollock2018operational,figueroa2019almost,Berk2021resourcetheoriesof,figueroa2021randomized,berk2023extracting,taranto2024characterising,zambon2024relations}. This approach, which uses what we call Choi divergences of quantum combs, is convenient because it does not require optimizations over inputs but has the drawback of the distinguishability measures not being monotonic under noisy manipulations of the process. This implies, for example, that according to such measures, it is possible to increase the non-Markovianity of the process by means of transformations that do not create correlations between different time steps.

While this drawback does not necessarily compromise the results that employ these measures, it was most likely overlooked in earlier analyses. For example, Ref. \cite{pollock2018operational}, which introduced Choi divergences as non-Markovianity quantifiers, discussed that the distinguishability measure between Choi states must be contractive to lead to consistent quantifiers, but the fact that such non-Markovianity quantifiers would still not be monotonic under local manipulations of the process was not considered for consistency. The data-processing inequality for these measures was discussed in only Refs. \cite{Berk2021resourcetheoriesof,berk2023extracting}, which, based on the contractivity of Choi divergences under superprocesses, concluded that their monotonicity holds.

Here, we begin by stating and identifying a gap in Claim 1, an implicit assumption from Ref. \cite{Berk2021resourcetheoriesof}, and reproduce from Ref. \cite{yuan2019hypothesis} a counterexample to it. Next, we show how this assumption propagates to situations pertaining only to process tensors, giving rise to Claims 2 and 3, both used in Ref. \cite{berk2023extracting}. We also give counterexamples to these claims. Then, we examine an alternative class of distinguishability measures for quantum combs, called generalized comb divergences. We prove some important properties of these divergences, including monotonicity under superprocesses. Finally, we discuss how this property indicates that generalized comb divergences are suitable distinguishability quantifiers for general applications of quantum combs, such as resource theories of quantum processes and quantum strategies \cite{Berk2021resourcetheoriesof,berk2023extracting,wang2019resource}.

\section{Quantum channel divergences}

We begin by analyzing the single-time scenario, described by quantum channels. Let $\Mcal_{\atb}$ be a quantum channel from an input space of density operators $\Dcal(\Hcal_{\rmA})$ to an output space $\Dcal(\Hcal_{\rmB})$. The (normalized) Choi state $\Upsilon_{\rmA\rmB}^{\Mcal}\in\Dcal(\Hcal_{\rmA}\otimes\Hcal_{\rmB})$ of $\Mcal_{\atb}$ is given by
\begin{equation}
\label{eq:choi-def}
    \Upsilon_{\rmA\rmB}^{\Mcal} \Def \qty(\Ical_{\rmA}\otimes \Mcal_{\atb})\Phi_{\rmA\rmA},
\end{equation}
where $\Phi_{\rmA\rmA} =\ket{\Phi}\bra{\Phi}_{\rmA\rmA}$,
\begin{equation}
\label{eq:phi}
    \ket{\Phi}_{\rmA\rmA} = \frac{1}{\sqrt{d}}\sum_{i=1}^d\ket{ii}_{\rmA\rmA}
\end{equation}
is a maximally entangled state in $\Hcal_{\rmA}\otimes\Hcal_{\rmA}$, and $\Ical_{\rmA}$ is the identity channel in $\Dcal(\Hcal_{\rmA})$. Note that $\Mcal_{\atb}$ being trace preserving implies that its Choi state satisfies
\begin{equation}
\label{eq:trace_cond}
    \Upsilon_{\rmA}^{\Mcal} = \Id_{\rmA},
\end{equation}
where $\Upsilon_{\rmA}^{\Mcal}=\tr_\rmB\qty[\Upsilon_{\rmA\rmB}^{\Mcal}]$ and $\Id=\Ibb/d$ is the maximally mixed state. This means that the set $\Ccal(\Hcal_{\rmA}\otimes\Hcal_{\rmB})$ of Choi states forms a strict subset of $\Dcal(\Hcal_{\rmA}\otimes\Hcal_{\rmB})$. Thus, quantum states $\rho\in\Dcal(\Hcal_{\rmA}\otimes\Hcal_{\rmB})$ with $\tr_\rmB\qty[\rho] \ne \Id_{\rmA}$ cannot be Choi states of channels. On the other hand, any quantum state $\Upsilon\in\Dcal(\Hcal_{\rmA}\otimes\Hcal_{\rmB})$ satisfying the above condition may be associated with a channel $\Mcal_{\atb}^{\Upsilon}:\Dcal(\Hcal_{\rmA})\to\Dcal(\Hcal_{\rmB})$ by
\begin{equation}
\label{eq:choi-back}
    \Mcal_{\atb}^{\Upsilon}(\rho_{\rmA}) = d\tr_\rmA\qty[\rho_{\rmA}\Upsilon_{\rmAB}^{T_\rmA}],
\end{equation}
where ${T_\rmA}$ is the partial transpose in the $\ket{i}$ basis of $\Hcal_{\rmA}$.

Now consider the task of determining how different two given quantum channels $\Mcal_{\atb}$ and $\Ncal_{\atb}$ are. One possible way to do this is by means of Choi divergences, which consist of applying state distinguishability measures to the Choi states $\Upsilon_{\rmA\rmB}^{\Mcal}$ and $\Upsilon_{\rmA\rmB}^{\Ncal}$ \cite{gour2021entropy}. These measures may be any generalized state divergence, defined as follows \cite{cooney2016strong,leditzky2018approaches}.
\begin{definition}[Generalized state divergences]
    A generalized state divergence $\bm{\rmD}\qty(\rho||\sigma)$ is a mapping from pairs of states to non-negative real numbers satisfying monotonicity under quantum channels $\bm{\rmD}\qty(\Mcal(\rho)||\Mcal(\sigma))\le\bm{\rmD}\qty(\rho||\sigma)$.
\end{definition}
Examples of generalized state divergences are the trace distance, $||\rho-\sigma||_1=\tr[|\rho-\sigma|]$, and the relative entropy, $S\qty(\rho||\sigma)=\tr[\rho(\log_2 \rho - \log_2\sigma)]$. We are now set to define Choi divergences.
\begin{definition}[Choi divergences of channels]
    A Choi divergence $\bm{\Tilde{\rmC}}\qty(\Mcal||\Ncal)$ between channels $\Mcal_{\atb}$ and $\Ncal_{\atb}$ is given by any generalized divergence between their Choi states,
\begin{equation}
\label{eq:choi-div}
    \bm{\Tilde{\rmC}}\qty(\Mcal||\Ncal) \Def \bm{\rmD}\qty(\Upsilon_{\rmA\rmB}^{\Mcal}||\Upsilon_{\rmA\rmB}^{\Ncal}).
\end{equation}
\end{definition}

As they do not require optimization over input states, Choi divergences are useful for applying properties of the state divergences and deriving relations for channels, like those of Ref. \cite{zambon2024relations}. Moreover, some interesting information quantifiers may be written as particular choices of Choi divergences. An example of this is the input-output correlation $M(\Mcal)$ of a channel 
$\Mcal$ used in Refs. \cite{berk2023extracting,zambon2024relations}, 
\begin{align}
        M(\Mcal)&=I(A:B)_{\Upsilon^{\Mcal}}\\
        &=S\qty(\Upsilon_{\rmA\rmB}^{\Mcal}||\Id_\rmA\otimes \Mcal_{\atb}(\Id_\rmA))\\
        &=\bm{\Tilde{\rmC}}\qty(\Mcal||\Ncal),
    \end{align}
where $\Ncal$ is a channel with $\sigma_\rmB=\Mcal_{\atb}(\Id_\rmA)$ as a fixed output. 

Despite these useful aspects, Choi divergences have the key disadvantage of not satisfying an important data-processing inequality. As discussed in Ref. \cite{Chiribella2008transforming}, the most general transformations from channels to channels are given by superchannels, which can always be implemented by means of preprocessing and postprocessing channels connected by an ancilla. In this sense, one would expect distinguishability measures of quantum channels to be contractive under the action of superchannels.

A common misconception here is to assume that the monotonicity of the state divergence $\bm{\rmD}\qty(\rho||\sigma)$ implies the monotonicity of the Choi divergence through Eq. \eqref{eq:choi-div} \cite{Berk2021resourcetheoriesof}. A possible reasoning to support this is the following.
\begin{claim}
    Consider that any superchannel $\Xi$ over channels $\Mcal_{\atb}$ induces a mapping $\tilde{\Xi}:\Ccal(\Hcal_{\rmA}\otimes\Hcal_{\rmB})\to\Ccal(\Hcal_{\rmA}\otimes\Hcal_{\rmB})$ over the set of Choi states. Since $\tilde{\Xi}$ maps states to states, it is a quantum channel in $\Ccal(\Hcal_{\rmA}\otimes\Hcal_{\rmB})$, and the monotonicity $\bm{\rmD}\qty(\tilde{\Xi}\qty(\Upsilon_{\rmA\rmB}^{\Mcal})||\tilde{\Xi}\qty(\Upsilon_{\rmA\rmB}^{\Ncal}))\le\bm{\rmD}\qty(\Upsilon_{\rmA\rmB}^{\Mcal}||\Upsilon_{\rmA\rmB}^{\Ncal})$ must hold.
\end{claim}
The problem with this reasoning is that despite $\tilde{\Xi}$ being linear and completely positive, it is not trace preserving or even trace nonincreasing in general. Even though it preserves the trace of states in $\Ccal(\Hcal_{\rmA}\otimes\Hcal_{\rmB})$, there may be linear operators over $\Hcal_{\rmA}\otimes\Hcal_{\rmB}$ for which it can even increase the trace. This implies that $\tilde{\Xi}$ is not a channel and the monotonicity does not hold in general even for divergences between two states in $\Ccal(\Hcal_{\rmA}\otimes\Hcal_{\rmB})$ \cite{muller2019monotonicity}. We now use a slightly modified version of an example from Ref. \cite{yuan2019hypothesis} to show the gap in Claim 1.

\begin{example}
    The detailed calculations for this example are carried out in Appendix \ref{app:ex1}. Consider a qubit channel $\Ncal_{\atb}$ with $\Id_\rmB$ as a fixed output, also called a completely depolarizing channel. Consider also a channel $\Mcal_{\atb}$ with Kraus operators $M_1=\sqrt{1/2}\ket{0}_\rmB\bra{0}_\rmA$, $M_2=\sqrt{1/2}\ket{1}_\rmB\bra{0}_\rmA$, and $M_3=\ket{1}_\rmB\bra{1}_\rmA$. The relative entropy between their Choi states is
    \begin{align}
            S\qty(\Upsilon_{\rmA\rmB}^{\Mcal}||\Upsilon_{\rmA\rmB}^{\Ncal})&= \frac{1}{2}.
        \end{align}
    
    Now consider the superchannel $\Xi\qty(\Ecal)=\Ecal\circ\Rcal$ in which $\Rcal_{\rmA\to\rmA}$ is a channel with $\ket{1}\bra{1}_\rmA$ as a fixed output. The channel $\tilde{\Xi}$ induced by $\Xi$ over the set of Choi states is given by
    \begin{equation}
        \tilde{\Xi}(\Upsilon_{\rmAB})= \Id_\rmA\otimes 2 \bra{1}_\rmA\Upsilon_{\rmAB}\ket{1}_\rmA,
    \end{equation}
    which can be directly verified through Eq. \eqref{eq:choi-back}. This implies $\tr[\tilde{\Xi}(\Upsilon_{\rmAB})]=\bra{1}\Upsilon_\rmA\ket{1}$, which together with Eq. \eqref{eq:trace_cond} ensures $\tilde{\Xi}$ preserves the trace of Choi states. This, however, does not mean that $\tilde{\Xi}$ is trace preserving or even trace nonincreasing in general, as required for the monotonicity of the relative entropy to hold \cite{muller2019monotonicity}. For example, for the linear operator $Q_{\rmAB}\Def(2\ket{1}\bra{1}_{\rmA}-\ket{0}\bra{0}_{\rmA})\otimes\Id_{\rmB}$ we have $\tr[Q_{\rmAB}]=1$ and $\tr[\tilde{\Xi}(Q_{\rmAB})]=2$.
    
    Finally, we notice that the relative entropy between states $\Upsilon_{\rmA\rmB}^{\Mcal}$ and $\Upsilon_{\rmA\rmB}^{\Ncal}$ increases under the action of $\tilde{\Xi}$, even though their traces are preserved, 
    \begin{align}
            S\qty(\tilde{\Xi}\qty(\Upsilon_{\rmA\rmB}^{\Mcal})||\tilde{\Xi}\qty(\Upsilon_{\rmA\rmB}^{\Ncal}))&= 1;
        \end{align}
    therefore, $S\qty(\tilde{\Xi}\qty(\Upsilon_{\rmA\rmB}^{\Mcal})||\tilde{\Xi}\qty(\Upsilon_{\rmA\rmB}^{\Ncal})) > S\qty(\Upsilon_{\rmA\rmB}^{\Mcal}||\Upsilon_{\rmA\rmB}^{\Ncal})$.
\end{example}

Since Choi divergences are not monotonic under superchannels, we also consider a different class of divergences, called generalized channel divergences \cite{leditzky2018approaches}.
\begin{definition}[Generalized channel divergences]
A generalized channel divergence $\bm{\Tilde{\rmD}}\qty(\Mcal||\Ncal)$ between channels $\Mcal_{\atb}$ and $\Ncal_{\atb}$ is given by
\begin{equation}
\label{eq:gen-div}
    \bm{\Tilde{\rmD}}\qty(\Mcal||\Ncal) \Def \sup_{\rho_{\rmR\rmA}}\bm{\rmD}\qty(\Mcal_{\atb}(\rho_{\rmR\rmA})||\Ncal_{\atb}(\rho_{\rmR\rmA})),
\end{equation}
where we have an optimization of a generalized state divergence between the outputs of the channels $\Ical_\rmR\otimes\Mcal_{\atb}$ and $\Ical_\rmR\otimes\Ncal_{\atb}$ given the same input, where $\rmR$ is an auxiliary system with $\dim\{\rmR\}\ge \dim\{\rmA\}$.
\end{definition}

Notice that the maximization in the definition above implies that generalized channel divergences may be hard to compute in general, especially when compared to Choi divergences, which require no optimization. Nevertheless, Refs. \cite{yuan2019hypothesis,gour2021entropy} showed that generalized channel divergences satisfy the most relevant properties expected for distinguishability quantifiers, including monotonicity under superchannels. This means that they are suitable measures for determining how different two given quantum channels are, which is especially relevant in the context of resource theories of quantum channels \cite{takagi2019general,gour2019comparison,gour2019how,liu2019resource,yuan2019hypothesis,yunchao2020operational,hsieh2020resource,hsieh2021communication,regula2021fundamental,stratton2024dynamical}. 

Importantly, the fact that Choi divergences are not monotonic under superchannels does not deem them useless. The Hilbert-Schmidt distance for quantum states, for example, is used in several bounds throughout quantum theory \cite{wilde2013quantum}, even though it is not contractive under quantum channels \cite{wang2009contractivity}. Similarly, Choi divergences may provide upper and lower bounds to generalized divergences, as shown in Remark 14 of Ref. \cite{katariya2021geometric} (see also Ref. \cite{wilde2020stack}).

\section{Process tensor divergences}

We now present the multitime scenario. Here, an experimenter performs control operations on a system of interest, which interacts with an uncontrolled environment between operations. The control operations may consist of an initial-state preparation and a sequence of quantum channels on the system alone, for example. The dynamics is then described by a process tensor $\Tcal$ mapping the control sequence $\Scal$ to the final state of the system \cite{pollock2018non}. Considering that the control operations may be correlated through an ancilla, the most general structure of both $\Tcal$ and $\Scal$ is that of quantum combs \cite{chiribella2009theoretical}, as shown in Fig. \ref{fig:combs}. To formalize this idea, let $\Pbb_n$ be the set of $n$-step quantum combs and $\Sbb_n$ be the set of combs that are mapped to system-ancilla states by the elements of $\Pbb_n$. In our description, process tensors are given by combs in $\Pbb_n$ and control sequences by combs in $\Sbb_n$. In Fig. \ref{fig:combs}, for example, we have $\Tcal\in\Pbb_2$ and $\Scal\in\Sbb_2$.

\begin{figure}[t]
    \centering
    \includegraphics[width=\columnwidth]{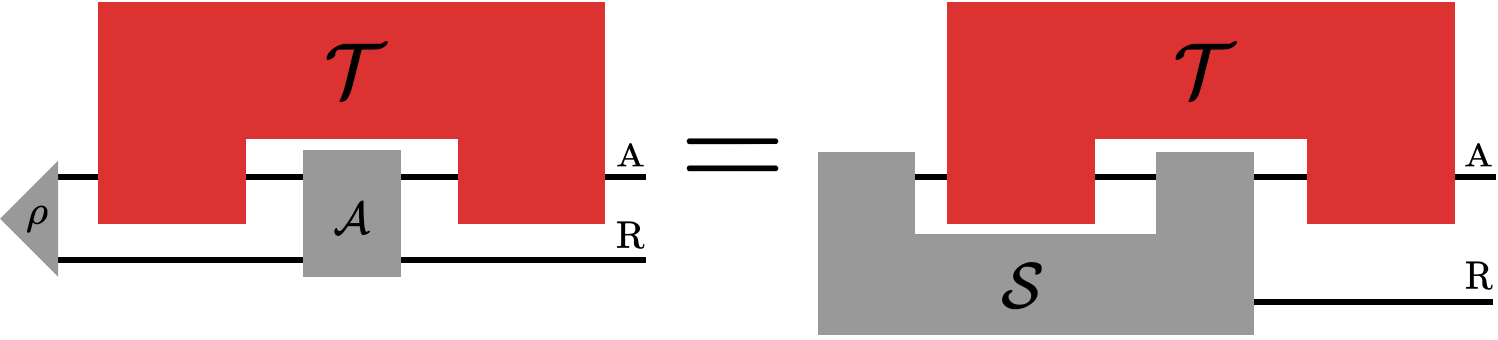}
    \caption{The two-step process tensor $\Tcal$ is a quantum comb that maps the control sequence to the final state of the system. It encapsulates the information of everything that is not controlled by the experimenter, namely, the initial state of the environment and the two interactions between the system and the environment. The most general control sequence for this case is given by a comb $\Scal$ containing the initial state $\rho\in\Dcal(\Hcal_\rmA\otimes\Hcal_\rmR)$ of system $\rmA$ and ancilla $\rmR$ and the operation $\Acal:\Dcal(\Hcal_\rmA\otimes\Hcal_\rmR)\to\Dcal(\Hcal_\rmA\otimes\Hcal_\rmR)$.}
    \label{fig:combs}
\end{figure}

Furthermore, one could consider an alternative situation where the experimenter has no control over the initial state of the system, which could even be correlated with the environment \cite{pollock2018non}. The process tensor associated with this dynamics would then carry information about the initial system-environment state as well as their subsequent interactions, mapping only the control operations to the final state of the system. Although we do not directly approach this second type of description here, all of our discussions apply to it with only small modifications.

Importantly, having a comb structure implies process tensors may also be associated with Choi states. This is done by inputting half a maximally entangled state at each step of the process and keeping the other halves and the outputs. A circuit for implementing this for a two-step process tensor is shown in Fig. \ref{fig:choi-circ}. We can also define combs $\Scal_{\text{Choi}}\in\Sbb_n$ implementing these circuits for $n$-step processes, such that $\Tcal(\Scal_{\text{Choi}})=\Upsilon^{\Tcal}$ for all $\Tcal\in\Pbb_n$. Similar to Eq. \eqref{eq:choi-back} for channels, there is also a way to obtain the action of combs from their Choi states. For combs $\Tcal\in\Pbb_n$ and $\Scal\in\Sbb_n$ with Choi states $\Upsilon^\Tcal$ and $\Upsilon^\Scal$, the action of $\Tcal$ on $\Scal$ is given by the \textit{link product} \cite{chiribella2009theoretical},
    \begin{align}
            \Upsilon^\Tcal\star\Upsilon^\Scal &\Def d_{\cap}\tr_{\cap}[\Upsilon^\Tcal\Upsilon^{\Scal^{T_\cap}}]\\
            &= \Tcal(\Scal),
        \end{align}
where $\cap$ is the intersection between the spaces over which the Choi states are defined, $d_{\cap}$ is its dimension, and $T_\cap$ is the partial transpose in this space.

Moreover, Choi states of $n$-step process tensors must satisfy a set of $n$ conditions, like that of Eq. \eqref{eq:trace_cond} for channels \cite{chiribella2009theoretical}. These conditions reflect not only the fact that the process tensor is trace preserving but also that it is time ordered, in the sense that future inputs cannot affect past outputs. For the Choi state $\Upsilon_{\rmA\rmB\rmC\rmD}^{\Tcal}$ of a two-step process tensor these conditions are
\begin{align}
    \Upsilon_{\rmA\rmB\rmC}^{\Tcal} &= \Upsilon_{\rmA\rmB}^{\Tcal}\otimes \Id_{\rmC}, \label{eq:order1} \\
    \Upsilon_{\rmA}^{\Tcal} &= \Id_{\rmA}. \label{eq:order2}
\end{align}
This means that, like in the channel case, the set $\Ccal(\Hcal_\rmA\otimes\Hcal_\rmB\otimes\Hcal_\rmC\otimes\Hcal_\rmD)$ of Choi states of two-step process tensors forms a strict subset of all states $\Dcal(\Hcal_\rmA\otimes\Hcal_\rmB\otimes\Hcal_\rmC\otimes\Hcal_\rmD)$.

\begin{figure}
    \centering
    \includegraphics[width=\columnwidth]{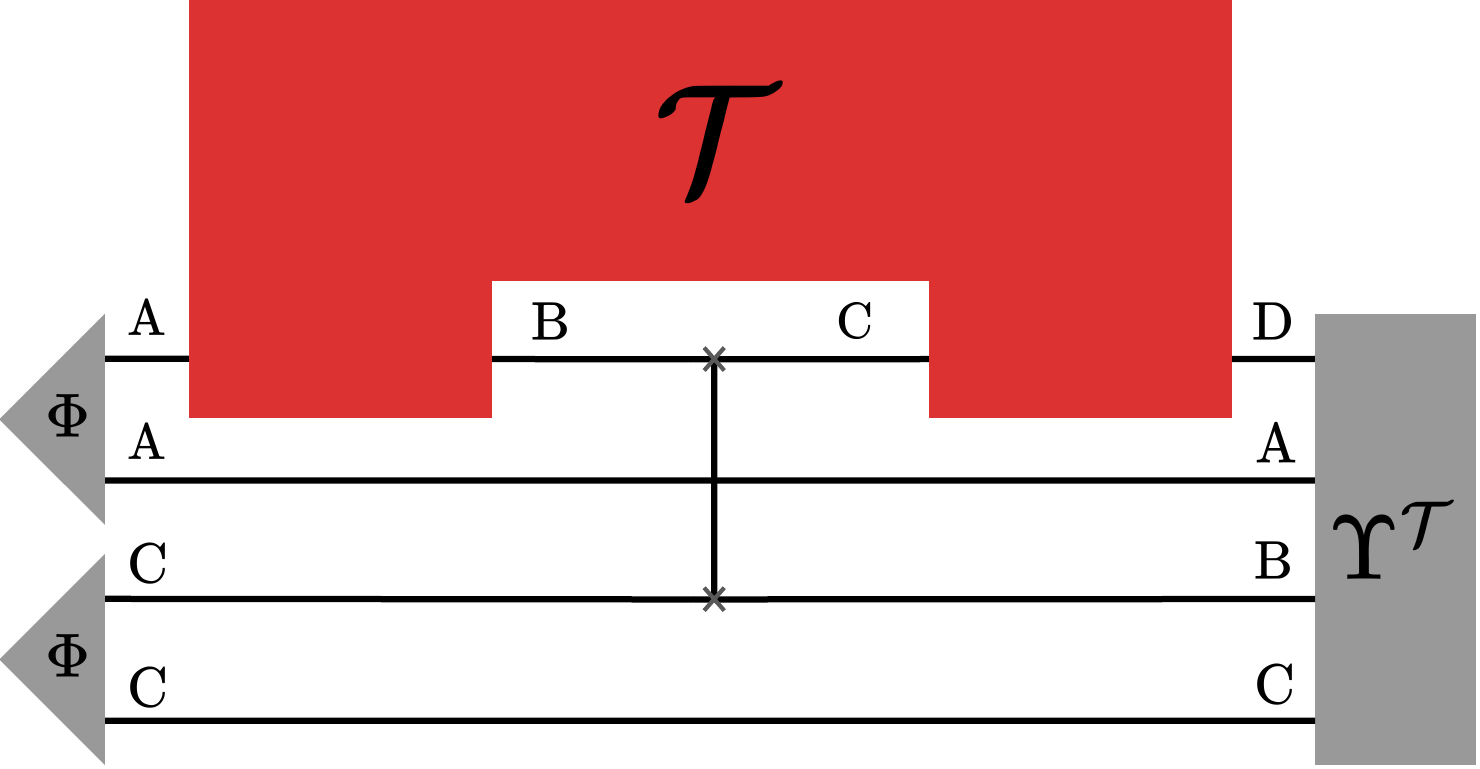}
    \caption{Circuit for obtaining the Choi state $\Upsilon^\Tcal$ of the process tensor $\Tcal$. First, we prepare the maximally entangled state $\Phi_{\rmA\rmA}$, input half of it to the first step of the process, and store the other half. Then, we prepare a second maximally entangled state $\Phi_{\rmC\rmC}$ and use a SWAP operator to input half of it to the second step of the process and store both the other half and the output of the first step. Finally, we store the output of the second step of the process. The four-partite stored state $\Upsilon^\Tcal\in\Dcal(\Hcal_\rmA\otimes\Hcal_\rmB\otimes\Hcal_\rmC\otimes\Hcal_\rmD)$ is the Choi state of the process tensor $\Tcal$.}
    \label{fig:choi-circ}
\end{figure}

Now, we consider the task of distinguishing two process tensors $\Tcal,\Vcal\in\Pbb_n$. We begin with a generalization of Choi divergences to quantum combs. 
\begin{definition}[Choi divergences of combs]
    A Choi divergence $\bm{\bar{\rmC}}\qty(\Tcal||\Vcal)$ between quantum combs $\Tcal,\Vcal\in\Pbb_n$ is given by any generalized state divergence between their Choi states
\begin{equation}
\label{eq:choi-div-pt}
    \bm{\bar{\rmC}}\qty(\Tcal||\Vcal) \Def \bm{\rmD}(\Upsilon^{\Tcal}||\Upsilon^{\Vcal}).
\end{equation}
\end{definition}

Given our discussion of the case of channels and considering that process tensors are multitime generalizations of channels, we would expect Choi divergences of process tensors to have the same problem as those of quantum channels. Indeed, we show next that they are not contractive under the action of superprocesses, the most general mappings from process tensors to process tensors \cite{Berk2021resourcetheoriesof}. Moreover, we discuss the consequences of this non-monotonicity for operations and properties that pertain only to the multitime case, like temporal coarse graining and non-Markovianity.

\begin{figure*}
    \centering
    \includegraphics[width=0.75\textwidth]{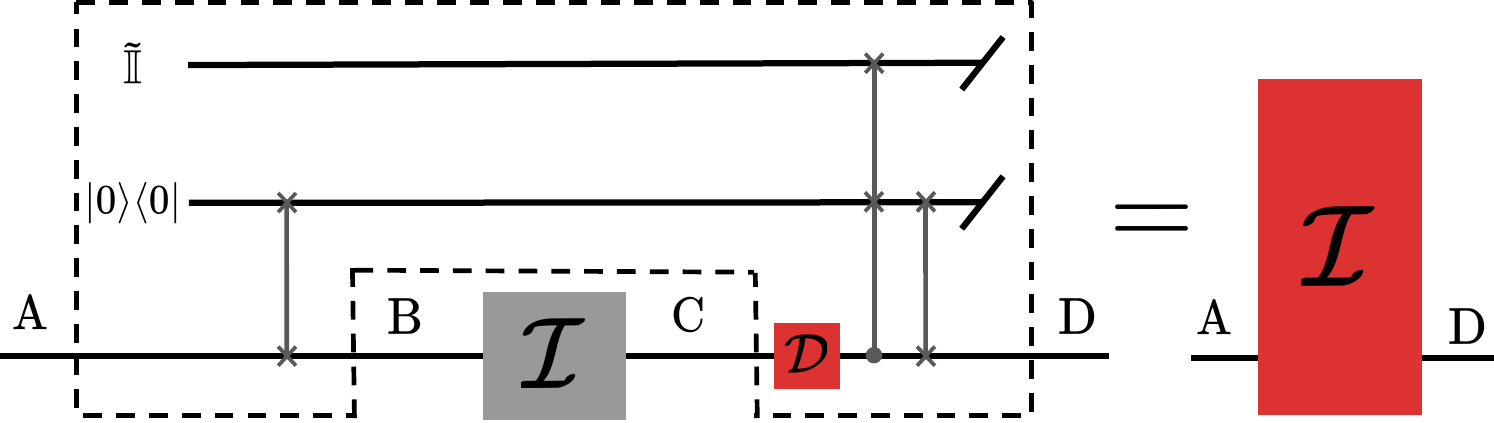}
    \caption{Coarse graining of a two-step process tensor $\Tcal$ resulting in a channel from the first input to the second output. The two-step process may be described as follows. In the first step, a SWAP gate exchanges the first input, some qubit state $\rho$, with the environment initial state $\ket{0}$, such that the first output is always $\ket{0}$ and $\rho$ is stored as the new state of the environment. In the second step, there is first the action of a fully dephasing channel $\Dcal(\cdot)=\bra{0}\cdot\ket{0}~\ket{0}\bra{0}+\bra{1}\cdot\ket{1}~\ket{1}\bra{1}$, then a controlled SWAP followed by another SWAP. This implies that if the second input of the dynamics is $\ket{1}$, the second output will be the maximally mixed state $\Id$. However, if the second input of the dynamics is $\ket{0}$, state $\rho$ is swapped back into the system, such that the second output will be exactly equal to the first input. Coarse graining this process means using the first output as the second input. Since the first output is always $\ket{0}$, the second output of the process will always be $\rho$. Therefore, the resulting channel after coarse graining is an identity channel from the first input to the second output.}
    \label{fig:ex2}
\end{figure*}

\subsection{Choi divergences under temporal coarse graining}

One could generalize to process tensors the idea presented for channels in Claim 1.

\begin{claim}\cite{berk2023extracting}
    Consider that any superprocess $\bm{Z}$ over process tensors $\Tcal$ induces a mapping $\tilde{\bm{Z}}$ over the set of Choi states. Since $\tilde{\bm{Z}}$ maps states to states, it is a quantum channel in the set of Choi states and the monotonicity $\bm{\rmD}\qty(\tilde{\bm{Z}}\qty(\Upsilon^{\Tcal})||\tilde{\bm{Z}}\qty(\Upsilon^{\Vcal}))\le\bm{\rmD}\qty(\Upsilon^{\Tcal}||\Upsilon^{\Vcal})$ must hold.
\end{claim}

From Example 1 we know the above statement cannot be true, as quantum channels are a special case of process tensors. However, we now present a counterexample to Claim 2 with an operation that is possible only in the multitime case.

\begin{example}
    Consider the action of a superprocess $\bm{G}:\Pbb_2\to\Pbb_1$ known as temporal coarse graining, which consists of doing nothing (inputting an identity channel $\Ical$) between two consecutive steps of the process. This is shown in Fig. \ref{fig:ex2}, where a two-step process tensor is mapped to a channel by means of coarse graining.

    Notice that in Fig. \ref{fig:ex2} the two-step process tensor $\Tcal$ has a Choi state $\Upsilon^{\Tcal}=\ket{0}\bra{0}_{\rmB}\otimes(\ket{0}\bra{0}_{\rmC}\otimes\Phi_{\rmA\rmD}+\ket{1}\bra{1}_{\rmC}\otimes\Id_{\rmA}\otimes\Id_{\rmD})/2$. Consider also the marginal process $\Tcal^{\text{marg}}$, whose Choi state is the uncorrelated version of $\Upsilon^{\Tcal}$, that is,
    \begin{align}
            \Upsilon^{\Tcal^{\text{marg}}}&=\Upsilon_\rmA^{\Tcal}\otimes\Upsilon_\rmB^{\Tcal}\otimes\Upsilon_\rmC^{\Tcal}\otimes\Upsilon_\rmD^{\Tcal}\\
            &=\Id_\rmA\otimes\ket{0}\bra{0}_{\rmB}\otimes\Id_\rmC\otimes\Id_\rmD.
        \end{align}
    The Choi divergence given by the relative entropy between $\Upsilon^{\Tcal}$ and $\Upsilon^{\Tcal^{\text{marg}}}$ is
    \begin{equation}
        I(\Tcal)\Def S(\Upsilon^{\Tcal}||\Upsilon^{\Tcal^{\text{marg}}}),
    \end{equation}
    which was defined in Ref. \cite{berk2023extracting} as the total correlations $I$ of the process $\Tcal$. We now show that for the process $\Tcal$ above $I$ is not contractive under the action of the temporal coarse-graining superprocess $\bm{G}$.
    
    First, in Appendix \ref{app:ex2} we obtain $I(\Tcal)= 1$. Note that while $\bm{G}(\Tcal)$ is an identity channel, as discussed in Fig. \ref{fig:ex2}, $\bm{G}(\Tcal)^{\text{marg}}$ is a channel with $\Id_\rmD$ as a fixed output, which implies
    \begin{align}
            I(\bm{G}(\Tcal))&=S(\Upsilon^{\bm{G}(\Tcal)}||\Upsilon^{\bm{G}(\Tcal)^{\text{marg}}})\\
            &= S(\Phi_{\rmA\rmD}||\Id_\rmA\otimes\Id_\rmD)\\
            &= 2.
        \end{align}
    Therefore, we have shown that the total correlations quantifier $I$ is not monotonic under the action of temporal coarse graining. Also, since in this particular case $\bm{G}(\Tcal)^{\text{marg}}=\bm{G}(\Tcal^{\text{marg}})$, we have 
    \begin{equation}
        S(\Upsilon^{\Tcal}||\Upsilon^{\Tcal^{\text{marg}}})<S(\Upsilon^{\bm{G}(\Tcal)}||\Upsilon^{\bm{G}(\Tcal^{\text{marg}})}),
    \end{equation}
    which shows that Choi divergences of process tensors are not contractive under superprocesses in general.
\end{example}

\subsection{Choi divergences as non-Markovianity quantifiers}

One of the key advantages of using process tensors to describe the dynamics of open quantum systems is that it provides an unambiguous definition of quantum Markovianity \cite{pollock2018non,pollock2018operational}. Namely, a process tensor is Markovian if each step of the dynamics consists of a quantum channel uncorrelated with the rest, which is equivalent to its Choi state being of product form. Such a definition allows for a proper quantification of the non-Markovianity of any given quantum process.

However, the non-Markovianity quantifiers employed so far consist of Choi divergences between the given process and the closest Markovian process \cite{pollock2018operational,figueroa2019almost,Berk2021resourcetheoriesof,figueroa2021randomized,berk2023extracting,taranto2024characterising,zambon2024relations}. Besides the geometrical aspect of measuring closeness to Markovian processes, the use of such divergences could be motivated by its monotonicity under superprocesses that do not create correlations between different time steps. These superprocesses, called independent quantum instruments (IQI) in Ref. \cite{berk2023extracting}, have the form of uncorrelated local superchannels acting on each step of the process, as shown in Fig. \ref{fig:mkv}. Extending the arguments of Claims 1 and 2, a possible reasoning to support Choi divergences as non-Markovianity quantifiers could be structured as follows.

\begin{claim}\cite{berk2023extracting}
    Any superprocess $\bm{Z}\in \text{IQI}$ induces a local mapping $\tilde{\bm{Z}}$ over the set of Choi states. Since local processing does not increase global correlations, it follows that 
    \begin{equation}
        \min_{\Vcal\in\Ccal^{\text{prod}}} \bm{\rmD}\qty(\tilde{\bm{Z}}\qty(\Upsilon^{\Tcal})|| \tilde{\bm{Z}}\qty(\Upsilon^{\Vcal})) \le\min_{\Vcal\in\Ccal^{\text{prod}}}\bm{\rmD}\qty(\Upsilon^{\Tcal}|| \Upsilon^{\Vcal}),
    \end{equation}
    where $\Ccal^{\text{prod}}$ is the set of product Choi states.
\end{claim}

Notice that ``local processing does not increase global correlations'' is generally implied by the monotonicity of the correlation measure under the processing \cite{wilde2013quantum}. However, from Example 2 we know this monotonicity does not hold for Choi divergences of process tensors, already putting Claim 3 into question. To show Claim 3 is not valid, we now provide an explicit counterexample to it.

\begin{example}
    Consider the non-Markovianity quantifier $N$, defined as
    \begin{equation}
        N(\Tcal)\Def \min_{\Vcal\in\Ccal^{\text{prod}}}S(\Upsilon^{\Tcal}|| \Upsilon^{\Vcal}).
    \end{equation}
    Let $\bm{Z}\in \text{IQI}$ be the superprocess consisting of preprocessing the second step with a channel that has $\ket{0}$ as a fixed output. 
    
    For the process tensor $\Tcal$ from Fig. \ref{fig:ex2}, in Appendix \ref{app:ex3} we calculate $N(\Tcal)= 1$. Note that $\Upsilon^{\bm{Z}(\Tcal)}=\Phi_{\rmA\rmD}\otimes\Id_{\rmB\rmC}$. Again in Appendix \ref{app:ex3} we show $N(\bm{Z}(\Tcal))= 2$. Therefore, the non-Markovianity quantifier $N$ may increase under the action of the local superprocess $\bm{Z}$.
\end{example}

\section{Generalized comb divergences}

As in the channels case, we turn to generalized divergences to obtain process tensor distinguishability measures that are contractive under superprocesses. To this end, we follow the definitions and results of Ref. \cite{wang2019resource}, where the quantum combs are called quantum strategies.
\begin{definition}[Generalized comb divergences]
A generalized comb divergence $\bm{\bar{\rmD}}\qty(\Tcal||\Vcal)$ between combs $\Tcal,\Vcal\in\Pbb_n$ is given by
\begin{equation}
\label{eq:pt-div}
    \bm{\bar{\rmD}}\qty(\Tcal||\Vcal) \Def \sup_{\Scal\in\Sbb_n}\bm{\rmD}\qty(\Tcal(\Scal)||\Vcal(\Scal)),
\end{equation}
where we have an optimization of a generalized state divergence between the outputs of the combs $\Tcal$ and $\Vcal$ given the same input.
\end{definition}
These divergences were shown to be contractive under the action of combs in Ref. \cite{wang2019resource}. We now generalize this result to contractivity under superprocesses.

\begin{theorem}
    Let $\bm{Z}$ be a general superprocess acting on process tensors $\Tcal,\Vcal\in\Pbb_n$. Then,
    \begin{equation}
    \label{eq:pt-monot}
        \bm{\bar{\rmD}}\qty(\bm{Z}(\Tcal)||\bm{Z}(\Vcal)) \le \bm{\bar{\rmD}}\qty(\Tcal||\Vcal).
\end{equation}
\end{theorem}
\begin{proof}
    In the most general case we have a superprocess $\bm{Z}:\Pbb_n\to\Pbb_m$. Reference \cite{berk2023extracting} showed that superprocesses have a dual action for contractions of compatible combs. This means that for any superprocess $\bm{Z}$ there is a dual superprocess $\bm{Z}^\dagger:\Sbb_m\to\Sbb_n$  such that $[\bm{Z}(\Tcal)](\Scal)=\Tcal(\bm{Z}^\dagger(\Scal))$ for all $\Scal\in\Sbb_m$. Let $\bm{Z}^\dagger[\Sbb_m]$ be the image of the set $\Sbb_m$ under the action of $\bm{Z}^\dagger$. Since $\bm{Z}^\dagger$ maps combs in $\Sbb_m$ to combs in $\Sbb_n$, we know that $\bm{Z}^\dagger[\Sbb_m]\subseteq\Sbb_n$. This implies
    \begin{align}
            \bm{\bar{\rmD}}\qty(\bm{Z}(\Tcal)||\bm{Z}(\Vcal)) &= \sup_{\Scal\in\Sbb_m}\bm{\rmD}\qty([\bm{Z}(\Tcal)](\Scal)||[\bm{Z}(\Vcal)](\Scal))\\
            &= \sup_{\Scal\in\Sbb_m}\bm{\rmD}\qty(\Tcal(\bm{Z}^\dagger(\Scal))||\Vcal(\bm{Z}^\dagger(\Scal)))\\
            &= \sup_{\Rcal\in\bm{Z}^\dagger[\Sbb_m]}\bm{\rmD}\qty(\Tcal(\Rcal)||\Vcal(\Rcal))\\
            &\le \sup_{\Rcal\in\Sbb_n}\bm{\rmD}\qty(\Tcal(\Rcal)||\Vcal(\Rcal))\\
            &= \bm{\bar{\rmD}}\qty(\Tcal||\Vcal).
        \end{align}
\end{proof}

\begin{figure}[t]
    \centering
    \includegraphics[width=\columnwidth]{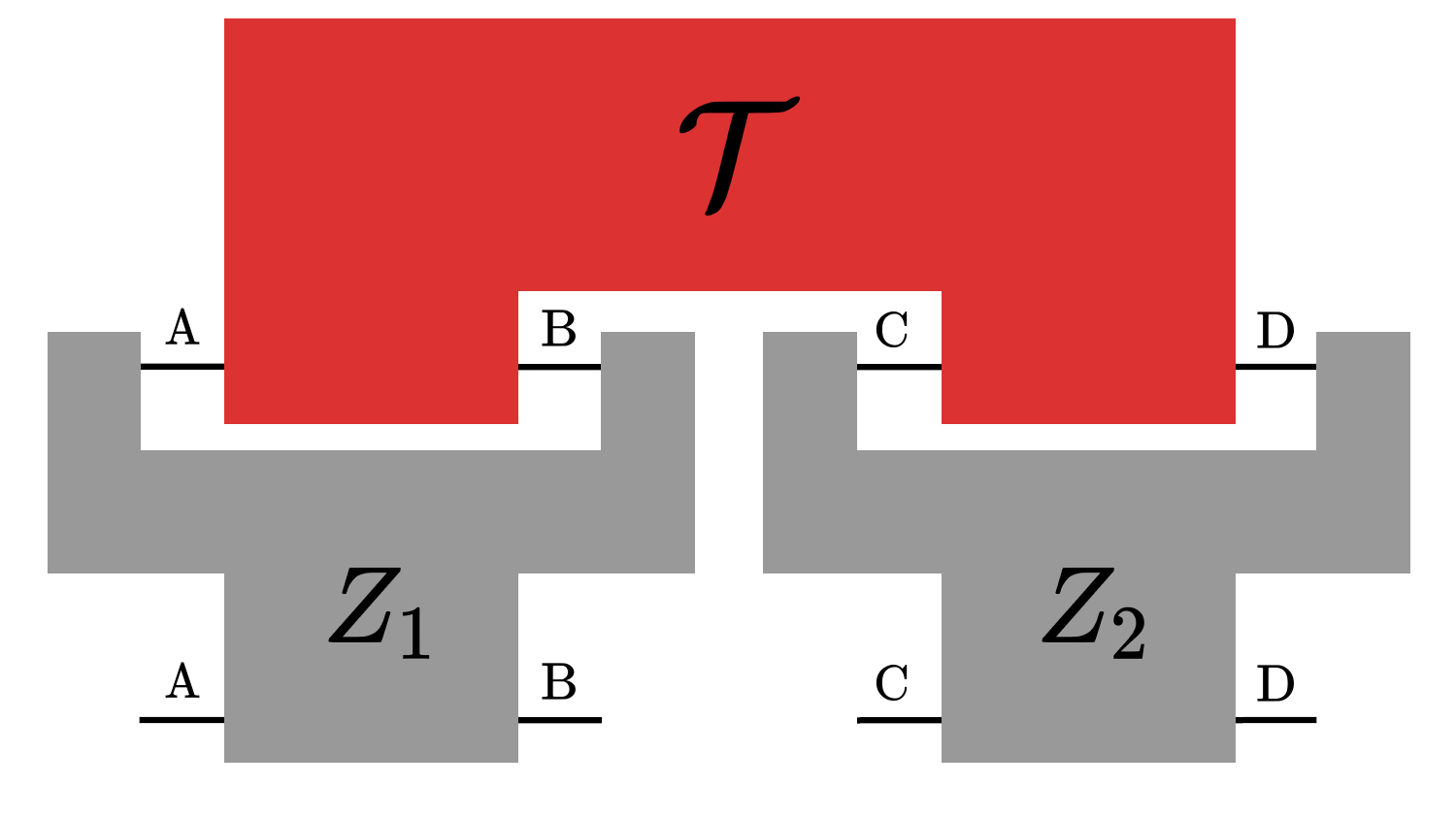}
    \caption{Action of a superprocess $\bm{Z}=\bm{Z}_1\otimes\bm{Z}_2$ on a process tensor $\Tcal$. $\bm{Z}_1$ and $\bm{Z}_2$ are uncorrelated superchannels, which implies $\bm{Z}\in\text{IQI}$. This kind of superprocess does not create correlations between different time steps; therefore, if $\Tcal$ is Markovian, then $\bm{Z}(\Tcal)$ will definitely also be.}
    \label{fig:mkv}
\end{figure}

Noteworthily, similar divergences have been used for process tensors \cite{figueroa2020equilibration,dowling2023relaxation,dowling2023equilibration,taranto2021non}, but they all consider a slightly different class of divergences, which we call classical generalized comb divergences. Before introducing them, we define classical generalized divergences.
\begin{definition}[Classical generalized divergences]
    A classical generalized divergence $D_C\qty(\bm{p}||\bm{q})$ is a mapping from pairs of probability distributions to non-negative real numbers satisfying monotonicity under noisy classical channels $D_C\qty(\Ncal(\bm{p})||\Ncal(\bm{q}))\le D_C\qty(\bm{p}||\bm{q})$.
\end{definition}
Examples of classical generalized divergences are the Kullback-Leibler divergence and the trace distance between probability distributions \cite{wilde2013quantum}. Now, we also consider the set of all objects $\Pcal\in\Tbb_n$ mapping combs in $\Pbb_n$ to probability distributions. Such objects are called \textit{testers} \cite{chiribella2009theoretical}. We then define classical generalized comb divergences.
\begin{definition}[Classical generalized comb divergences]
    A classical generalized comb divergence $\bm{\bar{\rmD}}_C\qty(\Tcal||\Vcal)$ between combs $\Tcal,\Vcal\in\Pbb_n$ is given by
    \begin{equation}
        \label{eq:comb-div}
        \bm{\bar{\rmD}}_C\qty(\Tcal||\Vcal) \Def \sup_{\Pcal\in\Tbb_n}D_C\qty(\Pcal(\Tcal)||\Pcal(\Vcal)),
    \end{equation}
    where we have an optimization of a classical generalized divergence between the output probability distributions $\bm{t}=\Pcal(\Tcal)$ and $\bm{v}=\Pcal(\Vcal)$ given the same tester.
\end{definition}
Using the trace distance as the classical divergence, we obtain \textit{distance between quantum combs} defined in Ref. \cite{chiribella2009theoretical} or \textit{strategy r-norms} from Ref. \cite{gutoski2012measure}.

Reference \cite{chiribella2009theoretical} also showed that the action of a tester $\Pcal\in\Tbb_n$ on a comb $\Tcal\in\Pbb_n$ can always be realized by some positive operator-valued measure (POVM) $P$ performed on state $\Tcal(\Scal_{\Pcal})$ for some $\Scal_{\Pcal}\in\Sbb_n$. We now show how this implies the contractivity of classical generalized comb divergences under superprocesses.
\begin{theorem}
    Let $\bm{Z}$ be a general superprocess acting on process tensors $\Tcal,\Vcal\in\Pbb_n$. Then,
    \begin{equation}
    \label{eq:classic-pt-monot}
        \bm{\bar{\rmD}}_C\qty(\bm{Z}(\Tcal)||\bm{Z}(\Vcal)) \le \bm{\bar{\rmD}}_C\qty(\Tcal||\Vcal).
\end{equation}
\end{theorem}
\begin{proof}
For every tester $\Pcal\in\Tbb_m$ realizable through a comb $\Scal_{\Pcal}\in\Sbb_m$ and a POVM $P=\{P_x\}$, let $\tilde{P}$ be the mapping from the final state $\rho$ to the probability distribution $\bm{p}$ obtained from measuring $P$, i.e., $p(x)=\tr[P_x\rho]$, such that $\Pcal(\Tcal)=\tilde{P}(\Tcal(\Scal_{\Pcal}))$. Then, for every superprocess $\bm{Z}:\Pbb_n\to\Pbb_m$ we have
\begin{align}
        \Pcal(\bm{Z}(\Tcal))&= \tilde{P}(\Tcal(\bm{Z}^\dagger(\Scal_{\Pcal})))\\
        &=\Pcal^\prime(\Tcal),
    \end{align}
where $\Pcal^\prime\in\Tbb_n$ is a tester realizable through the comb $\bm{Z}^\dagger(\Scal_{\Pcal})\in\Sbb_n$ and the same POVM $P$. This implies
    \begin{align}
            \bm{\bar{\rmD}}_C\qty(\bm{Z}(\Tcal)||\bm{Z}(\Vcal)) &= \sup_{\Pcal\in\Tbb_m}D_C\qty(\Pcal(\bm{Z}(\Tcal))||\Pcal(\bm{Z}(\Vcal)))\\
            &= \sup_{\Pcal^\prime\in\bm{Z}^\dagger[\Sbb_m]}D_C\qty(\Pcal^\prime(\Tcal)||\Pcal^\prime(\Vcal))\\
            &\le \sup_{\Pcal^\prime\in\Sbb_n}D_C\qty(\Pcal^\prime(\Tcal)||\Pcal^\prime(\Vcal))\\
            &= \bm{\bar{\rmD}}_C\qty(\Tcal||\Vcal).
        \end{align}
\end{proof}

Like in the channel case, Choi divergences of quantum combs
are not useless quantifiers for not being contractive. The fact that they are easily computable without requiring optimization over inputs makes them helpful tools in understanding relevant properties of quantum processes \cite{berk2023extracting,zambon2024relations}. Also, since there is always a comb $\Scal_{\text{Choi}}\in\Sbb_n$ such that $\Tcal(\Scal_{\text{Choi}})=\Upsilon^{\Tcal}$ for all $\Tcal\in\Pbb_n$, it is immediately clear that $\bm{\bar{\rmC}}(\Tcal||\Vcal)\le \bm{\bar{\rmD}}(\Tcal||\Vcal)$. Moreover, we now generalize the result of Ref. \cite{wilde2020stack} to the multitime scenario, showing that Choi divergences of quantum combs may also be used to establish an upper bound to generalized comb divergences.

\begin{theorem}
    Let $\bm{\rmD} $ be a generalized state divergence satisfying the direct-sum property,
    \begin{equation}
        \bm{\rmD} (\rho_{\rmX\rmA}^1||\rho_{\rmX\rmA}^2)= \sum_x p_x \bm{\rmD} (\rho_{\rmA}^1||\rho_{\rmA}^2)
    \end{equation}
    for $\rho_{\rmX\rmA}^i=\sum_x p_x \ket{x}\bra{x}_{\rmX}\otimes\rho_{\rmA}^i$. Then, for combs $\Tcal,\Vcal\in\Pbb_n$ with input dimension $d$ it holds that
    \begin{equation}
        \bm{\bar{\rmD}}(\Tcal||\Vcal)\le d^{2n-1} \bm{\bar{\rmC}}(\Tcal||\Vcal),
    \end{equation}
    where both the generalized comb divergence $\bm{\bar{\rmD}}$ and the Choi divergence of combs $\bm{\bar{\rmC}}$ are defined using the same state divergence $\bm{\rmD}$.
\end{theorem}
\begin{proof}
    Consider that for combs $\Tcal\in\Pbb_n$ and $\Scal\in\Sbb_n$ we have $\Tcal(\Scal)\in\Dcal(\Hcal_\rmA\otimes\Hcal_\rmR)$, where $\rmA$ is the output space of $\Tcal$ and $\rmR$ is the output space of $\Scal$ (see Fig. \ref{fig:combs}). Then, we define the Choi states $\Upsilon_{\rmR\rmA}^\Tcal\in\Dcal(\Hcal_\rmR\otimes\Hcal_\rmA)$ and $\Upsilon_{\rmR\rmR}^\Scal\in\Dcal(\Hcal_\rmR\otimes\Hcal_\rmR)$. Although $\Hcal_\rmR$ and $\Hcal_\rmA$ are not the usual input and output spaces over which we define Choi states, they are isomorphic to them (for $d\Def\dim\{\rmA\}$, we have $d_\rmR\Def\dim\{\rmR\}\le d^{2n-1}$) \cite{chiribella2009theoretical}. Now, define a channel $\Zcal:\Dcal(\Hcal_\rmR)\to\Dcal(\Hcal_\rmX\otimes\Hcal_\rmR)$ as
    \begin{align}
            \Zcal(\omega_\rmR) &\Def \frac{1}{d^{2n-1}}\ket{0}\bra{0}_\rmX\otimes \Upsilon_{\rmR\rmR}^\Scal\star\omega_\rmR\\
            &+ \ket{1}\bra{1}_\rmX\otimes \qty(\Id_{\rmR\rmR}-\frac{1}{d^{2n-1}}\Upsilon_{\rmR\rmR}^\Scal)\star\omega_\rmR.
        \end{align}
    The trace preservability of the channel is immediate from the definition, and its linearity follows from the linearity of the link product. To ensure it is completely positive, we use the fact that the link product of positive operators is positive \cite{chiribella2009theoretical} and that $\Id_{\rmR\rmR}-\Upsilon_{\rmR\rmR}^\Scal/d^{2n-1}$ is always positive (this is, in fact, the only reason to add the factor $1/d^{2n-1}$ to the definition).
    
    By construction, it follows that
    \begin{align}
            \Zcal(\Upsilon_{\rmR\rmA}^\Tcal) &= \frac{1}{d^{2n-1}}\ket{0}\bra{0}_\rmX\otimes \Tcal(\Scal)\\
            &+\qty(1-\frac{1}{d^{2n-1}})\ket{1}\bra{1}_\rmX\otimes \eta_{\rmR\rmA},
        \end{align}
    where $\eta_{\rmR\rmA}$ is some irrelevant state. The above equation means that by acting only on $\rmR$ the channel $\Zcal$ steers the Choi state $\Upsilon_{\rmR\rmA}^\Tcal$ of the comb $\Tcal$ to the state $\Tcal(\Scal)$ with probability $1/d^{2n-1}$.

    Now we use the contractivity of the generalized state divergence under the channel $\Zcal$ and the direct-sum property to show
    \begin{align}
            \bm{\bar{\rmC}}(\Tcal||\Vcal) &= \bm{\rmD} (\Upsilon_{\rmR\rmA}^\Tcal||\Upsilon_{\rmR\rmA}^\Vcal)\\
            &\ge \bm{\rmD} (\Zcal(\Upsilon_{\rmR\rmA}^\Tcal)||\Zcal(\Upsilon_{\rmR\rmA}^\Vcal))\\
            &= \frac{1}{d^{2n-1}}\bm{\rmD} (\Tcal(\Scal)||\Vcal(\Scal))\\
            &+\qty(1-\frac{1}{d^{2n-1}})\bm{\rmD} (\eta_{\rmR\rmA}||\nu_{\rmR\rmA})\\
            &\ge \frac{1}{d^{2n-1}}\bm{\rmD} (\Tcal(\Scal)||\Vcal(\Scal)).
        \end{align}
    Since this holds for any $\Scal\in\Sbb_n$, we have
    \begin{align}
            \bm{\bar{\rmD}} (\Tcal||\Vcal) &= \sup_{\Scal\in\Sbb_n}\bm{\rmD}\qty(\Tcal(\Scal)||\Vcal(\Scal))\\
            &\le d^{2n-1} \bm{\bar{\rmC}}(\Tcal||\Vcal).
        \end{align}
\end{proof}

\section{Conclusions}

We discussed Choi divergences and generalized divergences as distinguishability measures for process tensors. Starting from the channel case, we exposed the gap in a claim that implies the monotonicity of Choi divergences under superchannels, and then we reproduced an explicit counterexample to it. Next, we showed what problems arise when this wrong claim is extended to process tensors. Namely, correlation quantifiers defined through Choi divergences do not satisfy contractivity under superprocesses, which implies, for example, that a widely used non-Markovianity quantifier may increase under the action of superprocesses that are local in time.

To circumvent these issues we turned to generalized divergences, which are known to satisfy important properties for quantum channels, including monotonicity under superchannels. We then analyzed its generalizations to quantum combs, namely, generalized comb divergences and classical generalized comb divergences, for which we proved the monotonicity under superprocesses. Finally, we generalized from channels to combs a relation between the two classes of divergences, allowing one to use Choi divergences to establish lower and upper bounds to generalized divergences.

A few questions naturally arise at this point. For example, how are the results of Refs. \cite{pollock2018operational,figueroa2019almost,Berk2021resourcetheoriesof,figueroa2021randomized,berk2023extracting,taranto2024characterising,zambon2024relations} impacted by our results, given that they use Choi divergences as non-Markovianity quantifiers? Also, if one defines a non-Markovianity quantifier as a generalized comb divergence to the closest Markov process, how does this relate to operational quantities, like the maximum amount of correlations present in the output of the dynamics under general combs? With such a definition, would it be possible to establish bounds like those of Ref. \cite{zambon2024relations}?

Besides these questions regarding general applications of distinguishability measures, the consequences of our results are clear in the context of resource theories whose states have a comb structure, like resource theories of process tensors and quantum strategies. For once, we could expect quantifiers defined using Choi divergences to, in general, violate monotonicity under free operations. On the other hand, quantifiers defined as generalized comb divergences to the closest free comb will certainly be monotonic under free operations of any comb resource theory. Therefore, we expect our results to clarify some misconceptions and provide a viable alternative approach to process tensor distinguishability measures, hopefully leading to a better understanding of important features of multitime quantum process, especially those relevant to quantum information-processing tasks.

\begin{acknowledgments}
I thank D. O. Soares-Pinto, G. Adesso, S. Milz, G. Berk, and K. Modi for insightful discussions. I acknowledge support from the S{\~a}o Paulo State Foundation (FAPESP) under Grants No. 2022/00993-9 and No. 2023/04625-7.
\end{acknowledgments}

\appendix

\section{Example 1}\label{app:ex1}

Consider a qubit channel $\Ncal_{\atb}$ with $\Id_\rmB$ as a fixed output, also called a completely depolarizing channel. The Choi state for this channel is
\begin{align}
    \Upsilon_{\rmA\rmB}^{\Ncal}&= (\Ical_{\rmA}\otimes\Ncal_{\atb})\Phi_{\rmA\rmA}\\
    &= \Id_\rmA\otimes\Id_\rmB.
\end{align}

Consider also a channel $\Mcal_{\atb}$ with Kraus operators $M_1=\sqrt{1/2}\ket{0}_\rmB\bra{0}_\rmA$, $M_2=\sqrt{1/2}\ket{1}_\rmB\bra{0}_\rmA$, and $M_3=\ket{1}_\rmB\bra{1}_\rmA$. Its Choi state is given by
\begin{align}
    \Upsilon_{\rmA\rmB}^{\Mcal}&= (\Ical_{\rmA}\otimes\Mcal_{\atb})\Phi_{\rmA\rmA}\\
    &= \frac{1}{2}\sum_{i,j=0}^1 \ket{i}\bra{j}_{\rmA}\otimes \sum_{k=1}^3 M_k\ket{i}\bra{j}_\rmA M_k^\dagger\\
    &= \frac{1}{4}\qty(\ket{00}\bra{00}_{\rmAB}+\ket{01}\bra{01}_{\rmAB}+2\ket{11}\bra{11}_{\rmAB}).\label{eq:eigen}
\end{align}

Now we calculate the relative entropy between them
\begin{align}
            S\qty(\Upsilon_{\rmA\rmB}^{\Mcal}||\Upsilon_{\rmA\rmB}^{\Ncal})&= S\qty(\Upsilon_{\rmA\rmB}^{\Mcal}||\Id_\rmA\otimes\Id_\rmB)\\
            &= 2 + S\qty(\Upsilon_{\rmA\rmB}^{\Mcal}||\Ibb_\rmA\otimes\Ibb_\rmB)\\
            &= 2 - H\qty(\Upsilon_{\rmA\rmB}^{\Mcal}),
\end{align}
in which we first used $S(\rho||a\sigma)=\log_2 a+S(\rho||\sigma)$ and $\log_2 2=1$ and then $S(\rho||\Ibb)=-H(\rho)$, where $H\qty(\rho)=-\tr[\rho\log_2 \rho]$ is the von Neumann entropy of $\rho$. Given we know from Eq. \eqref{eq:eigen} that $\Upsilon_{\rmA\rmB}^{\Mcal}$ has eigenvalues $\bm{\lambda}=\{1/4,1/4,1/2\}$, we have 
\begin{align}
    H\qty(\Upsilon_{\rmA\rmB}^{\Mcal})&=-\sum\lambda_i\log_2\lambda_i\\
    &= -\qty(\frac{1}{4}\log_2\frac{1}{4}+\frac{1}{4}\log_2\frac{1}{4}+\frac{1}{2}\log_2\frac{1}{2})\\
    &= \frac{3}{2},
\end{align}
implying
\begin{equation}
    S\qty(\Upsilon_{\rmA\rmB}^{\Mcal}||\Upsilon_{\rmA\rmB}^{\Ncal})=\frac{1}{2}.
\end{equation}

Now consider the superchannel $\Xi\qty(\Ecal)=\Ecal\circ\Rcal$ in which $\Rcal_{\rmA\to\rmA}$ is a channel with $\ket{1}\bra{1}_\rmA$ as a fixed output. The channel $\tilde{\Xi}$ induced by $\Xi$ over the set of Choi states is given by
\begin{equation}
        \tilde{\Xi}(\Upsilon_{\rmAB})= \Id_\rmA\otimes 2 \bra{1}_\rmA\Upsilon_{\rmAB}\ket{1}_\rmA,
\end{equation}
which we directly verify through Eq. \eqref{eq:choi-back},
\begin{align}
    \Xi[\Mcal_{\atb}^{\Upsilon}](\rho_{\rmA}) &= \Mcal_{\atb}^{\Upsilon}(\mathcal{R}(\rho_{\rmA}))\\
    &= \Mcal_{\atb}^{\Upsilon}(\ket{1}\bra{1}_\rmA)\\
    &=  d\tr_\rmA\qty[\ket{1}\bra{1}_\rmA\Upsilon_{\rmAB}^{T_\rmA}]\\
    &= 2\bra{1}_\rmA\Upsilon_{\rmAB}\ket{1}_\rmA\\
    &= d \tr_\rmA\qty[\rho_\rmA(\Id_\rmA\otimes 2\bra{1}_\rmA\Upsilon_{\rmAB}\ket{1}_\rmA)^{T_\rmA}]\\
    &= d\tr_\rmA\qty[\rho_\rmA\tilde{\Xi}(\Upsilon_{\rmAB})^{T_\rmA}].
\end{align} 

Notice that the channel $\Xi\qty(\Mcal)$ has $\Mcal_{\atb}(\ket{1}\bra{1}_\rmA)=\ket{1}\bra{1}_\rmB$ as a fixed output, while $\Xi\qty(\Ncal)$ has $\Ncal_{\atb}(\ket{1}\bra{1}_\rmA)=\Id_\rmB$ as a fixed output. This implies their Choi states are $\Upsilon_{\rmA\rmB}^{\Xi(\Mcal)}=\Id_\rmA\otimes\ket{1}\bra{1}_\rmB$ and $\Upsilon_{\rmA\rmB}^{\Xi(\Ncal)}=\Id_\rmA\otimes\Id_\rmB$. The relative entropy between them is
\begin{align}
            S\qty(\tilde{\Xi}\qty(\Upsilon_{\rmA\rmB}^{\Mcal})||\tilde{\Xi}\qty(\Upsilon_{\rmA\rmB}^{\Ncal}))&= S\qty(\Upsilon_{\rmA\rmB}^{\Xi(\Mcal)}||\Upsilon_{\rmA\rmB}^{\Xi(\Ncal)})\\
            &= S\qty(\Id_\rmA\otimes\ket{1}\bra{1}_\rmB||\Id_\rmA\otimes\Id_\rmB)\\
            &= S\qty(\Id_\rmA||\Id_\rmA)+S\qty(\ket{1}\bra{1}_\rmB||\Id_\rmB)\\
            &= 1;
\end{align}
therefore, $S\qty(\tilde{\Xi}\qty(\Upsilon_{\rmA\rmB}^{\Mcal})||\tilde{\Xi}\qty(\Upsilon_{\rmA\rmB}^{\Ncal})) > S\qty(\Upsilon_{\rmA\rmB}^{\Mcal}||\Upsilon_{\rmA\rmB}^{\Ncal})$.

\section{Example 2}\label{app:ex2}

For the process with Choi state $\Upsilon_{\rmAB\rmC\rmD}^{\Tcal}=\ket{0}\bra{0}_{\rmB}\otimes(\ket{0}\bra{0}_{\rmC}\otimes\Phi_{\rmA\rmD}+\ket{1}\bra{1}_{\rmC}\otimes\Id_{\rmA}\otimes\Id_{\rmD})/2$ we have
\begin{align}
    I(\Tcal)&=S(\Upsilon_{\rmAB\rmC\rmD}^{\Tcal}||\Upsilon_{\rmAB\rmC\rmD}^{\Tcal^{\text{marg}}})\\
        &= S(\Upsilon_{\rmA\rmC\rmD}^{\Tcal}||\Upsilon_{\rmA\rmC\rmD}^{\Tcal^{\text{marg}}})\\
        &= -H(\Upsilon_{\rmA\rmC\rmD}^{\Tcal})-\tr[\Upsilon_{\rmA\rmC\rmD}^{\Tcal}\log_2\Id_{\rmA\rmC\rmD}],
\end{align}
where we first used the fact that $\Upsilon_{\rmAB\rmC\rmD}^{\Tcal}=\ket{0}\bra{0}_{\rmB}\otimes\Upsilon_{\rmA\rmC\rmD}^{\Tcal}$ implies $S(\Upsilon_{\rmAB\rmC\rmD}^{\Tcal}||\Upsilon_{\rmAB\rmC\rmD}^{\Tcal^{\text{marg}}})=S(\Upsilon_{\rmA\rmC\rmD}^{\Tcal}||\Upsilon_{\rmA\rmC\rmD}^{\Tcal^{\text{marg}}})$ and then $S(\rho||\sigma)=-H(\rho)+\tr[\rho\log_2\sigma]$. Now, to calculate $H(\Upsilon_{\rmA\rmC\rmD}^{\Tcal})$ we use the fact that the entropy of a classical-quantum state is \cite{wilde2013quantum}
\begin{equation}
    H\qty(\sum_ip_i\ket{i}\bra{i}\otimes\rho_i) = -\sum_i p_i\log_2 p_i+\sum_ip_iH(\rho_i),
\end{equation}
which implies
\begin{align}
    H(\Upsilon_{\rmA\rmC\rmD}^{\Tcal})&= H\qty(\frac{1}{2}\ket{0}\bra{0}_{\rmC}\otimes\Phi_{\rmA\rmD}+\frac{1}{2}\ket{1}\bra{1}_{\rmC}\otimes\Id_{\rmA}\otimes\Id_{\rmD})\\
    &= 1 +\frac{1}{2}H(\Phi_{\rmA\rmD})+\frac{1}{2}H(\Id_{\rmA}\otimes\Id_{\rmD})\\
    &= 2
\end{align}
since $H(\Phi_{\rmA\rmD})=0$ for a pure state and $H(\Id_{\rmA}\otimes\Id_{\rmD})=2$. For the other term we have
\begin{align}
    \tr[\Upsilon_{\rmA\rmC\rmD}^{\Tcal}\log_2\Id_{\rmA\rmC\rmD}]&= \tr[\Upsilon_{\rmA\rmC\rmD}^{\Tcal}\qty(\log_2\Ibb_{\rmA\rmC\rmD}-3)]\\
    &= -3,
\end{align}
where we used $\log_2\Ibb_{\rmA\rmC\rmD}=0$ (null operator) and $\tr[\Upsilon_{\rmA\rmC\rmD}^{\Tcal}]=1$. Finally, we obtain
\begin{align}
    I(\Tcal)&= -H(\Upsilon_{\rmA\rmC\rmD}^{\Tcal})-\tr[\Upsilon_{\rmA\rmC\rmD}^{\Tcal}\log_2\Id_{\rmA\rmC\rmD}]\\
    &= 1.
\end{align}

\section{Example 3}\label{app:ex3}

The non-Markovianity quantifier is given by
\begin{align}
    N(\Tcal)&= \min_{\Vcal\in\Ccal^{\text{prod}}}S(\Upsilon_{\rmAB\rmC\rmD}^{\Tcal}|| \Upsilon_{\rmAB\rmC\rmD}^{\Vcal})\\
    &= S(\Upsilon^{\Tcal}|| \Upsilon_{\rmAB}^{\Tcal}\otimes\Upsilon_{\rmC\rmD}^{\Tcal})\\
    &=H(\Upsilon_{\rmAB}^{\Tcal})+H(\Upsilon_{\rmC\rmD}^{\Tcal})-H(\Upsilon_{\rmAB\rmC\rmD}^{\Tcal}),
\end{align}
where we first used the fact that the uncorrelated state closest to any state in relative entropy is the product of its marginals and then the definition of the quantum mutual information \cite{wilde2013quantum}. Since the Choi state is $\Upsilon_{\rmAB\rmC\rmD}^{\Tcal}=\ket{0}\bra{0}_{\rmB}\otimes(\ket{0}\bra{0}_{\rmC}\otimes\Phi_{\rmA\rmD}+\ket{1}\bra{1}_{\rmC}\otimes\Id_{\rmA}\otimes\Id_{\rmD})/2$, we have $\Upsilon_{\rmAB}^{\Tcal}=\Id_{\rmA}\otimes\ket{0}\bra{0}_{\rmB}$ and $\Upsilon_{\rmC\rmD}^{\Tcal}=\Id_{\rmC\rmD}$, implying
\begin{align}
    N(\Tcal)&= H(\Id_{\rmA}\otimes\ket{0}\bra{0}_{\rmB})+H(\Id_{\rmC\rmD})-H(\Upsilon_{\rmAB\rmC\rmD}^{\Tcal})\\
    &= 1+2-2\\
    &= 1,
\end{align}
where we used $H(\Upsilon_{\rmAB\rmC\rmD}^{\Tcal})=2$ from Appendix \ref{app:ex2}.

Given the action of the superprocess $\bm{Z}$, the Choi state is $\Upsilon_{\rmAB\rmC\rmD}^{\bm{Z}(\Tcal)}=\Phi_{\rmA\rmD}\otimes\Id_{\rmB\rmC}$. Then we have $\Upsilon_{\rmAB}^{\Tcal}=\Id_{\rmAB}$ and $\Upsilon_{\rmC\rmD}^{\Tcal}=\Id_{\rmC\rmD}$, implying
\begin{align}
    N(\bm{Z}(\Tcal))&= H(\Id_{\rmAB})+H(\Id_{\rmC\rmD})-H(\Phi_{\rmA\rmD}\otimes\Id_{\rmB\rmC})\\
    &= 2 + 2 - 2\\
    &= 2.
\end{align}


%

\end{document}